\documentclass[10pt,twocolumn]{article}

\usepackage[T1]{fontenc}
\usepackage[utf8]{inputenc}
\usepackage{amsmath,amssymb,amsthm}
\usepackage{graphicx}
\usepackage{booktabs}
\usepackage{url}
\usepackage{hyperref}
\usepackage{textcomp}
\usepackage[margin=0.75in,top=0.75in,bottom=1in,columnsep=0.25in]{geometry}
\usepackage{titlesec}
\usepackage{caption}
\usepackage{cite}

\titleformat{\section}{\centering\normalfont\scshape}{\Roman{section}.}{0.5em}{}
\titleformat{\subsection}{\normalfont\itshape}{\Alph{subsection}.}{0.5em}{}
\titleformat{\subsubsection}{\normalfont\itshape}{\arabic{subsubsection})}{0.5em}{}
\titlespacing*{\section}{0pt}{12pt plus 2pt minus 2pt}{6pt plus 1pt minus 1pt}
\titlespacing*{\subsection}{0pt}{8pt plus 2pt minus 2pt}{4pt plus 1pt minus 1pt}
\titlespacing*{\subsubsection}{0pt}{6pt plus 2pt minus 2pt}{3pt plus 1pt minus 1pt}

\hypersetup{colorlinks=true,linkcolor=black,citecolor=black,urlcolor=blue}

\newtheorem{theorem}{Theorem}
\newtheorem{proposition}[theorem]{Proposition}
\newtheorem{corollary}[theorem]{Corollary}
\newtheorem{definition}[theorem]{Definition}
\newtheorem{assumption}{Assumption}

\begin{document}

\twocolumn[
\begin{@twocolumnfalse}
\begin{center}
{\Large\bfseries Forecastability as an Information-Theoretic Limit on Prediction\par}
\vspace{12pt}
{\large Peter Maurice Catt\par}
\vspace{4pt}
{\normalsize Independent Researcher\par}
\vspace{2pt}
{\normalsize Auckland, New Zealand\par}
\vspace{18pt}
\end{center}
\end{@twocolumnfalse}
]

\noindent\textbf{\textit{Abstract}---}Forecasting is usually framed as a problem of model choice. This paper starts earlier, asking how much predictive information is available at each horizon. Under logarithmic loss, the answer is exact: the mutual information between the future observation and the declared information set equals the maximum achievable reduction in expected loss~\cite{delsole2004predictability}. This paper develops the consequences of that identity. Forecastability, defined as this mutual information evaluated across horizons, forms a profile whose shape reflects the dependence structure of the process and need not be monotone. Three structural properties are derived: compression of the information set can only reduce forecastability; the gap between the profile under a finite lag window and the full history gives an exact truncation error budget; and for processes with periodic dependence, the profile inherits the periodicity. Predictive loss decomposes into an irreducible component fixed by the information structure and an approximation component attributable to the method; their ratio defines the exploitation ratio, a normalised diagnostic for method adequacy. The exact equality is specific to log loss, but when forecastability is near zero, classical inequalities imply that no method under any loss can materially improve on the unconditional baseline. The framework provides a theoretical foundation for assessing, prior to any modelling, whether the declared information set contains sufficient predictive information at the horizon of interest.

\vspace{6pt}

\vspace{12pt}

\section{Introduction}

Forecasting is usually framed as a problem of model choice. This paper starts earlier. It asks how much predictive information the available data contain at each horizon, and what that implies for the best accuracy any method can achieve. The question has antecedents in Wiener's theory of prediction~\cite{wiener1949extrapolation} and Granger's characterisation of predictive relations~\cite{granger1969investigating}, but neither framework provides an exact, model-free, horizon-specific bound under a proper scoring rule.

The central result is that, under logarithmic loss, the mutual information $I(Y_{t+h}; \mathcal{I}_t)$ between the future observation and the declared information set is exactly the maximum achievable reduction in expected loss relative to the unconditional baseline. This identity is established in the predictability literature by DelSole~\cite{delsole2004predictability}; the present paper develops its consequences for forecasting practice. It is not an approximation or an asymptotic bound; it is an exact equality, achieved when the forecasting method coincides with the true conditional distribution. The quantity $I(Y_{t+h}; \mathcal{I}_t)$, evaluated across horizons, defines the forecastability profile: a horizon-resolved function that maps each forecast horizon to the maximum predictive improvement available from the declared information set.

Forecastability is not a scalar quantity but a structure over horizons. The shape of the profile, not any single number, determines where modelling effort is justified and where it is not.

Prior work has studied predictability using information-theoretic quantities, but in different forms and for different purposes. In climate science, DelSole~\cite{delsole2004predictability} and DelSole and Tippett~\cite{delsole2007predictability} characterise predictability within a given dynamical model, building on the signal-to-noise framework of Leith~\cite{leith1978predictability}---their framework conditions on a specified dynamical model, whereas the present paper conditions on a declared information set without assuming a model. Bialek et al.~\cite{bialek2001predictability} study predictive information as a complexity quantity over observation length, extending ideas from computational mechanics~\cite{crutchfield1989inferring,shalizi2001computational}; Goerg~\cite{goerg2013foreca} ranks components by spectral entropy; Pennekamp et al.~\cite{pennekamp2019intrinsic}, Ponce-Flores et al.~\cite{ponceflores2020complexity}, and Papacharalampous et al.~\cite{papacharalampous2022massive} use entropy-based measures descriptively to relate time series structure to realised forecast performance. Related model-free approaches include Garland et al.~\cite{garland2014modelfreequant} on permutation entropy-based quantification and Sugihara and May~\cite{sugihara1990nonlinear} on nonlinear forecasting as a diagnostic for chaos. None develops the identity into a systematic, model-free, horizon-resolved diagnostic framework conditional on a declared information set. The underlying information-theoretic results are standard~\cite{cover2006elements,kullback1959information}; the contribution is their assembly into a coherent framework for forecasting, together with the structural properties of the forecastability profile (Corollary~5, Propositions~6 and~7) and the conceptual vocabulary (forecastability profile, exploitability, exploitation ratio) needed to make them operational.

The theoretical foundations rest on the entropy and mutual information framework of classical information theory~\cite{shannon1948mathematical,fano1961transmission,renyi1961measures,cover2006elements,gray2011entropy}, the divergence-based approach to statistical discrimination~\cite{kullback1959information,csiszar2011information}, the theory of proper scoring rules and their connection to probabilistic forecasting~\cite{gneiting2007strictly}, and prior work establishing forecastability as an information-theoretic diagnostic for time series~\cite{catt2009forecastability,catt2014entropy,catt2026diagnostic}.

The paper is organised as follows. Section~II defines the setting. Section~III establishes the information bound. Section~IV introduces the forecastability profile and derives its structural properties, including the data processing inequality corollary, the finite-window information loss, and a periodicity result, illustrated by two Gaussian examples. Section~V gives the loss decomposition. Section~VI characterises behaviour at horizons of low forecastability. Section~VII discusses estimation. Section~VIII develops implications for practice. Section~IX concludes.

\section{Setting}

Let $\{Y_t\}_{t \in \mathbb{Z}}$ be a stochastic process with values in a measurable space $(\mathcal{Y}, \mathcal{B})$. At time~$t$, the forecaster uses a declared information set $\mathcal{I}_t$, represented as a sub-$\sigma$-algebra. All results are conditional on this choice: different information sets imply different limits on achievable performance.

The true predictive object is the conditional distribution $p_h^\star(y \mid \mathcal{I}_t) = P(Y_{t+h} \in dy \mid \mathcal{I}_t)$ at horizon $h \geq 1$. A forecasting method supplies an approximation $q_h(\cdot \mid \mathcal{I}_t)$. Performance is evaluated under logarithmic loss, $\ell(q, y) = -\log q(y)$. Logarithmic loss is used here not because it exhausts all forecasting objectives, but because it is the setting in which predictive information admits an exact and transparent decision-theoretic interpretation: it is strictly proper, equals cross-entropy in expectation, and has a direct coding interpretation~\cite{gneiting2007strictly,gneiting2014probabilistic,shannon1948mathematical,cover2006elements}. The connection between proper scoring rules and Bayesian decision theory is developed in Dawid~\cite{dawid1986probability} and the relationship between calibration and refinement in DeGroot and Fienberg~\cite{degroot1983comparison}.

The mutual information between the future and the information set is
\begin{equation}
I(Y_{t+h}; \mathcal{I}_t) = H(Y_{t+h}) - H(Y_{t+h} \mid \mathcal{I}_t),
\end{equation}
which measures the reduction in uncertainty about $Y_{t+h}$ obtainable from $\mathcal{I}_t$~\cite{cover2006elements,fano1961transmission}. The variation of this quantity across horizons constitutes the forecastability profile of the process. The information set may include lagged values, exogenous variables, covariates, or external signals; additional variables improve forecasting only insofar as they increase $I(Y_{t+h}; \mathcal{I}_t)$.

\subsection{Regularity conditions}

The following conditions are assumed throughout.

\begin{assumption}
(i)~All distributions admit densities with respect to a common dominating measure $\mu$ on $(\mathcal{Y}, \mathcal{B})$, and $q_h$ is absolutely continuous with respect to $p_h^\star$ wherever KL divergence is invoked. (ii)~The entropies $H(Y_{t+h})$, $H(Y_{t+h} \mid \mathcal{I}_t)$, and the expected KL divergence $\mathbb{E}[D_{\mathrm{KL}}(p_h^\star \| q_h)]$ are finite. (iii)~When the forecastability profile is treated as a function of horizon alone, the process $\{Y_t\}$ is stationary~\cite{doob1953stochastic}, so that $I(Y_{t+h}; \mathcal{I}_t)$ does not depend on~$t$.
\end{assumption}

Condition~(i) ensures that the log score and KL divergence are well-defined; the absolute continuity requirement follows standard conventions in information theory~\cite{pinsker1964information,billingsley1995probability}. In the continuous case, all entropies are differential entropies defined relative to the dominating measure~$\mu$; the mutual information $I(Y_{t+h}; \mathcal{I}_t)$ is invariant to the choice of $\mu$ and is therefore well-defined without ambiguity~\cite{cover2006elements}. Condition~(ii) excludes pathological cases where entropy is infinite. Condition~(iii) is sufficient for Proposition~4 and for the interpretation of the forecastability profile as a function of horizon alone. It is not necessary for the remaining results: all formal results other than Proposition~4 hold pointwise at each $(t, h)$ without stationarity, and require only that $I(Y_{t+h}; \mathcal{I}_t)$ is well-defined under conditions~(i) and~(ii). When the process is not stationary, the forecastability profile should be interpreted as local in time and indexed by both~$h$ and~$t$.

\section{The Information Bound on Prediction}

\begin{theorem}[Conditional-distribution optimality]
For any predictive distribution $q_h(\cdot \mid \mathcal{I}_t)$ satisfying Assumption~1,
\begin{equation}
\mathbb{E}[-\log q_h(Y_{t+h} \mid \mathcal{I}_t)] = H(Y_{t+h} \mid \mathcal{I}_t) + \mathbb{E}[D_{\mathrm{KL}}(p_h^\star \| q_h)].
\end{equation}
Consequently,
\begin{equation}
\mathbb{E}[-\log q_h(Y_{t+h} \mid \mathcal{I}_t)] \geq H(Y_{t+h} \mid \mathcal{I}_t),
\end{equation}
with equality if and only if $q_h = p_h^\star$ almost surely.
\end{theorem}

\begin{proof}
\begin{align}
&\mathbb{E}[-\log q_h(Y_{t+h} \mid \mathcal{I}_t)] \notag\\
&= \mathbb{E}[-\log p_h^\star(Y_{t+h} \mid \mathcal{I}_t)] + \mathbb{E}\!\left[\log \frac{p_h^\star(Y_{t+h} \mid \mathcal{I}_t)}{q_h(Y_{t+h} \mid \mathcal{I}_t)}\right] \notag\\
&= H(Y_{t+h} \mid \mathcal{I}_t) + \mathbb{E}[D_{\mathrm{KL}}(p_h^\star \| q_h)].
\end{align}
Since $D_{\mathrm{KL}}(p_h^\star \| q_h) \geq 0$ with equality iff $p_h^\star = q_h$ a.s., the bound follows.
\end{proof}

Theorem~1 identifies the benchmark: the best achievable loss is the conditional entropy, and every other method pays an additional KL penalty. The non-negativity of KL divergence, central to the proof, is a consequence of Jensen's inequality; its role as the unique divergence consistent with the log score was established by Bernardo~\cite{bernardo1979expected} and is reviewed in Gr\"{u}nwald and Dawid~\cite{grunwald2004gametheoretic}. The marginal predictor is the unconditional distribution $p(Y_{t+h})$, which ignores $\mathcal{I}_t$ and achieves expected log loss $H(Y_{t+h})$.

\begin{theorem}[Information bound on prediction]
The maximum achievable reduction in expected log loss from exploiting $\mathcal{I}_t$, relative to the marginal predictor, is exactly
\begin{equation}
H(Y_{t+h}) - H(Y_{t+h} \mid \mathcal{I}_t) = I(Y_{t+h}; \mathcal{I}_t).
\end{equation}
This bound is achieved if and only if $q_h = p_h^\star$ almost surely.
\end{theorem}

\begin{proof}
The marginal predictor achieves expected log loss $H(Y_{t+h})$. The minimum achievable loss using $\mathcal{I}_t$ is $H(Y_{t+h} \mid \mathcal{I}_t)$ by Theorem~1. The difference is $I(Y_{t+h}; \mathcal{I}_t)$ by~(1).
\end{proof}

The bound is tight, but generally unattainable in practice because $p_h^\star$ is unknown. The gap between the bound and achieved performance is the approximation error $\mathbb{E}[D_{\mathrm{KL}}(p_h^\star \| q_h)]$~\cite{kullback1951information}. Theorem~2 recovers, in the present notation, the identity established by DelSole~\cite{delsole2004predictability} in the context of climate predictability; the contribution of this paper lies not in the identity itself but in the structural consequences developed in Sections~IV--VI. Convergence rates for the approximation term under various model classes are studied in Barron~\cite{barron1991minimum} and Yang and Barron~\cite{yang1999information}. A direct consequence is the horizon-specific bound: for any method $q_h$,
\begin{equation}
H(Y_{t+h}) - \mathbb{E}[-\log q_h(Y_{t+h} \mid \mathcal{I}_t)] \leq F(h; \mathcal{I}_t),
\end{equation}
since the left-hand side equals $F(h; \mathcal{I}_t) - \mathbb{E}[D_{\mathrm{KL}}(p_h^\star \| q_h)] \leq F(h; \mathcal{I}_t)$ by non-negativity of KL divergence.

\section{Forecastability Profiles and Their Properties}

\subsection{Forecastability}

\begin{definition}[Forecastability]
The forecastability of the process $\{Y_t\}$ at horizon $h \geq 1$ relative to information set $\mathcal{I}_t$ is
\begin{equation}
F(h; \mathcal{I}_t) = I(Y_{t+h}; \mathcal{I}_t).
\end{equation}
\end{definition}

By Theorem~2, forecastability is the maximum number of nats by which any probabilistic method can improve on the marginal predictor at horizon~$h$, under logarithmic loss. It is defined relative to the declared information set and is therefore not an intrinsic property of the time series alone. Changes in the composition of $\mathcal{I}_t$ will, in general, change the forecastability profile. Hereafter, $F(h; \mathcal{I}_t)$ is used throughout; $I(Y_{t+h}; \mathcal{I}_t)$ appears only when the chain rule or data processing inequality is applied directly to the mutual information.

The profile should be interpreted as a horizon-resolved structure rather than a scalar quantity. Its shape reflects the temporal dependence in the data and may be non-monotone. Fraser and Swinney~\cite{fraser1986independent} introduced lagged mutual information for embedding dimension selection in dynamical systems; Li~\cite{li1990mutual} established the relationship between mutual information functions and correlation functions. In the present framework, the profile serves a different purpose: in many practical series, forecastability is concentrated at specific lags, particularly seasonal periods, rather than distributed smoothly across horizons.

\subsection{The entropy rate and maximum forecastability}

For a stationary process with the full history as information set, $\mathcal{I}_t = \sigma(Y_t, Y_{t-1}, \ldots)$, the forecastability at the shortest horizon has a direct interpretation in terms of the entropy rate.

\begin{proposition}[Maximum forecastability and entropy rate]
Let $h_\infty = \lim_{n \to \infty} H(Y_n \mid Y_{n-1}, \ldots, Y_1)$ denote the entropy rate of the stationary process $\{Y_t\}$ (assumed finite under Assumption~1), and let $\mathcal{I}_t = \sigma(Y_t, Y_{t-1}, \ldots)$. Then
\begin{equation}
F(1; \mathcal{I}_t) = H(Y_{t+1}) - h_\infty.
\end{equation}
\end{proposition}

\begin{proof}
For a stationary process, $H(Y_{t+1} \mid \mathcal{I}_t) = H(Y_{t+1} \mid Y_t, Y_{t-1}, \ldots) = h_\infty$ by definition of the entropy rate. Therefore $F(1; \mathcal{I}_t) = H(Y_{t+1}) - H(Y_{t+1} \mid \mathcal{I}_t) = H(Y_{t+1}) - h_\infty$.
\end{proof}

The entropy rate $h_\infty$ is the irreducible uncertainty per observation that remains when the full past is known~\cite{cover2006elements,gray2011entropy}. For stationary ergodic processes, the Shannon--McMillan--Breiman theorem~\cite{breiman1957individual,barron1985strong} ensures that sample averages of log-likelihood converge to the entropy rate almost surely. Hence $H(Y_{t+1}) - h_\infty$ gives the maximum one-step forecastability and anchors the profile at its shortest horizon.

\subsection{Information-set compression and the data processing inequality}

Any transformation or compression of the information set cannot increase forecastability. This is a direct consequence of the data processing inequality~\cite{cover2006elements}.

\begin{corollary}[Forecastability under information-set compression]
Let $T(\mathcal{I}_t)$ be any measurable function of the information set. Then
\begin{equation}
F(h; T(\mathcal{I}_t)) \leq F(h; \mathcal{I}_t)
\end{equation}
for all $h \geq 1$. Equality holds if and only if $T$ is a sufficient statistic for $Y_{t+h}$ given $\mathcal{I}_t$ in the sense of Fisher~\cite{fisher1922mathematical} and Halmos and Savage~\cite{halmos1949application}.
\end{corollary}

\begin{proof}
By the data processing inequality, $F(h; T(\mathcal{I}_t)) \leq F(h; \mathcal{I}_t)$ for any measurable~$T$. Equality holds iff the Markov chain $Y_{t+h} - \mathcal{I}_t - T(\mathcal{I}_t)$ satisfies the sufficiency condition~\cite{cover2006elements}.
\end{proof}

This result has immediate consequences: feature extraction, dimensionality reduction, lag truncation, and any other summarisation of the raw information set can only reduce the forecastability profile. The profile under the raw information set is therefore an upper bound on what any reduced representation can achieve.

\subsection{Finite-window information loss}

Even when the full history $\mathcal{I}_t = \sigma(Y_t, Y_{t-1}, \ldots)$ is available, forecasters typically condition on a finite lag window $\mathcal{I}_t^{(p)} = \sigma(Y_t, Y_{t-1}, \ldots, Y_{t-p+1})$ for some $p \geq 1$. The information lost by this truncation is itself a mutual information quantity.

\begin{proposition}[Finite-window information loss]
For a stationary process satisfying Assumption~1, the forecastability lost by truncating the information set from the full history to a $p$-lag window is
\begin{align}
&F(h; \mathcal{I}_t) - F(h; \mathcal{I}_t^{(p)}) \notag\\
&\quad= I(Y_{t+h}; Y_{t-p}, Y_{t-p-1}, \ldots \mid Y_t, \ldots, Y_{t-p+1}).
\end{align}
This quantity is non-negative and equals zero if and only if $Y_{t+h}$ is conditionally independent of the remote past given the $p$ most recent observations.
\end{proposition}

\begin{proof}
By the chain rule for mutual information,
\begin{align}
I(Y_{t+h}; \mathcal{I}_t) &= I(Y_{t+h}; Y_t, \ldots, Y_{t-p+1}) \notag\\
&\quad + I(Y_{t+h}; Y_{t-p}, Y_{t-p-1}, \ldots \mid Y_t, \ldots, Y_{t-p+1}).
\end{align}
The first term is $F(h; \mathcal{I}_t^{(p)})$. The second is non-negative by definition of conditional mutual information and equals zero iff the stated conditional independence holds.
\end{proof}

This gives the empirical programme an exact error budget: the gap between the estimated forecastability profile under a finite window and the true profile under the full history is a well-defined information-theoretic quantity. For processes with short memory, this gap vanishes for moderate~$p$. For processes with long-range dependence~\cite{beran1994statistics}, it may remain substantial. The choice of lag order $p$ connects to the embedding problem in nonlinear dynamics~\cite{tong1990nonlinear,kantz2004nonlinear}.

\subsection{Periodicity and non-monotonicity}

\begin{proposition}[Periodic forecastability]
Let $\{Y_t\}$ be a stationary process such that the conditional distribution $P(Y_{t+h} \mid \mathcal{I}_t)$ has period~$s$ in~$h$, in the sense that $P(Y_{t+h+s} \mid \mathcal{I}_t) = P(Y_{t+h} \mid \mathcal{I}_t)$ for all $h \geq 1$. Then $F(h+s; \mathcal{I}_t) = F(h; \mathcal{I}_t)$ for all $h \geq 1$.
\end{proposition}

\begin{proof}
Under the stated periodicity, $H(Y_{t+h+s} \mid \mathcal{I}_t) = H(Y_{t+h} \mid \mathcal{I}_t)$ and $H(Y_{t+h+s}) = H(Y_{t+h})$ by stationarity. Therefore $F(h+s; \mathcal{I}_t) = H(Y_{t+h+s}) - H(Y_{t+h+s} \mid \mathcal{I}_t) = H(Y_{t+h}) - H(Y_{t+h} \mid \mathcal{I}_t) = F(h; \mathcal{I}_t)$.
\end{proof}

The forecastability profile therefore need not be monotone: a process with seasonal period~$s$ will show recurring peaks at multiples of~$s$. The proof follows directly from the definitions; the substantive content is the implication that scalar diagnostics, which compress the dependence structure into a single number, cannot detect this horizon-resolved structure. More generally, even without exact periodicity, strong seasonal or cyclical dependence can produce local peaks in forecastability at seasonal lags superimposed on a long-run decay. Under standard mixing conditions~\cite{bradley2005basic,ibragimov1978gaussian} with sufficient regularity, forecastability converges to zero as $h \to \infty$~\cite{gray2011entropy}, but the path need not be monotone at any finite range of horizons. This non-monotonicity is precisely what distinguishes the horizon-resolved profile from scalar diagnostics such as sample entropy or spectral entropy, which cannot identify horizons where predictive structure re-emerges after an intervening gap.

\subsection{Interpretation of the forecastability profile}

The forecastability profile partitions the horizon space into regions where modelling effort has different expected returns. At horizons where forecastability is large, improvements in model specification can translate into measurable gains. At horizons where forecastability is small, the expected reduction in log loss achievable by any method is bounded by the forecastability itself, and improvements in model complexity yield diminishing returns. Rather than necessarily indicating model inadequacy, weak performance relative to the unconditional baseline may reflect the absence of exploitable structure in the declared information set. The profile therefore separates model-limited regimes from information-limited regimes.

\subsection{Illustrative examples}

Two Gaussian AR processes illustrate the forecastability profile and connect it to the structural properties derived above.

\textit{Example 1 (AR(1): monotone decay).} Let $Y_t = \phi Y_{t-1} + \varepsilon_t$ with $\varepsilon_t \sim \mathcal{N}(0, \sigma^2)$ and $|\phi| < 1$. The process is Markov, so $\mathcal{I}_t = \sigma(Y_t)$ is sufficient for the full history. The conditional distribution of $Y_{t+h}$ given $Y_t = y$ is $\mathcal{N}(\phi^h y,\; \sigma^2(1 - \phi^{2h})/(1 - \phi^2))$, and the marginal is $\mathcal{N}(0, \sigma^2/(1 - \phi^2))$. The forecastability profile is
\begin{equation}
F(h) = -\tfrac{1}{2}\log(1 - \phi^{2h}).
\end{equation}
Since the population coefficient of determination is $R_h^2 = \phi^{2h}$, this is equivalently $F(h) = -\frac{1}{2}\log(1 - R_h^2)$, an instance of the Gaussian relationship developed in Section~VIII.

Fig.~\ref{fig:profiles}(a) plots the profile for $\phi = 0.95$ and $\phi = 0.3$. For the weakly dependent process ($\phi = 0.3$), forecastability is $0.047$ nats at $h = 1$ and falls below $0.005$ nats by $h = 2$: no method can materially improve on the marginal predictor beyond the first step. For the strongly dependent process ($\phi = 0.95$), forecastability exceeds $0.5$ nats through $h = 10$ and remains above $0.1$ nats past $h = 30$, so modelling effort is justified across a wide horizon range. The contrast between the two profiles is the diagnostic in action: prior to fitting any model, the forecastability profile identifies where predictive improvement is available.

Since the AR(1) is Markov of order~1, the finite-window information loss (Proposition~6) satisfies $\Delta_p(h) = 0$ for all $p \geq 1$: a single lag captures all predictive information, and the truncation error budget is identically zero.

\textit{Example 2 (Seasonal AR: non-monotone profile).} Consider the multiplicative seasonal process
\begin{equation}
Y_t = \phi Y_{t-1} + \Phi Y_{t-s} - \phi\Phi Y_{t-s-1} + \varepsilon_t
\end{equation}
with $\phi = 0.5$, $\Phi = 0.8$, seasonal period $s = 12$, and $\varepsilon_t \sim \mathcal{N}(0, \sigma^2)$. This is the stationary Gaussian $\mathrm{AR}(1,0) \times (1,0)_{12}$ model. Consider first the restricted information set $\mathcal{I}_t = \sigma(Y_t)$, so that $F(h; \sigma(Y_t)) = -\frac{1}{2}\log(1 - \rho_h^2)$ where $\rho_h$ is the autocorrelation at lag~$h$.

Fig.~\ref{fig:profiles}(b) plots the resulting profile. It is non-monotone. Forecastability falls from $0.14$ nats at $h = 1$ to a minimum of approximately $0.0004$ nats near $h = 6$, where the AR(1) component has largely dissipated and the seasonal component has not yet engaged. It then rises sharply to $0.49$ nats at $h = 12$, where the annual dependence dominates. Further peaks appear at $h = 24$ ($0.25$ nats) and $h = 36$ ($0.13$ nats), attenuated by successive powers of~$\Phi$.

The profile partitions the horizon axis into qualitatively distinct regimes. At short horizons ($h = 1$--$3$), forecastability is dominated by the AR(1) component. At intermediate horizons ($h \approx 5$--$9$), the marginal predictor is nearly optimal and increased model complexity yields negligible returns. At seasonal horizons ($h \approx 12, 24, 36$), forecastability re-emerges; a method that fails to exploit the seasonal structure would show a low exploitation ratio $\chi_q$ at those horizons. This non-monotonicity is the phenomenon formalised in Proposition~7: scalar diagnostics that compress the dependence structure into a single number cannot detect this re-emergence.

The finite-window information loss (Proposition~6) is instructive here. The profile under $\mathcal{I}_t = \sigma(Y_t)$ misses all dependence beyond lag~1. Extending the information set to include the seasonal lag, $\mathcal{I}_t^{(13)} = \sigma(Y_t, \ldots, Y_{t-12})$, captures the seasonal structure and raises the profile at all horizons. The gap between the two profiles is the truncation budget $\Delta_p(h)$: it is negligible at short horizons (where the AR(1) component dominates) and large at seasonal horizons (where remote lags carry substantial conditional information). The error budget thus provides a principled criterion for lag-order selection.

For a non-Gaussian process with the same autocovariance structure, the forecastability profile would in general differ because mutual information captures dependence beyond second moments. By the maximum-entropy property of the Gaussian, the profile computed here is a lower bound on the forecastability of any process with identical autocovariances.

\begin{figure}[t]
\centering
\includegraphics[width=\columnwidth]{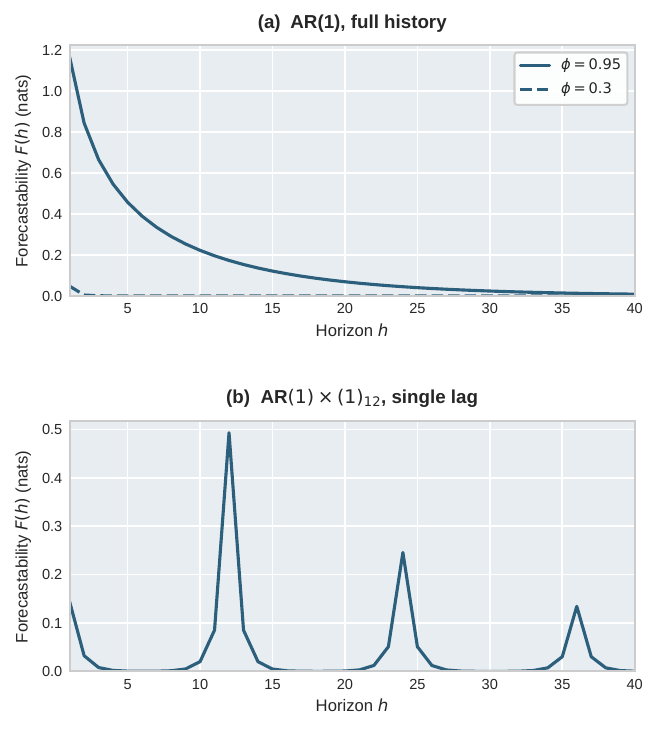}
\caption{Forecastability profiles for Gaussian processes. (a)~AR(1) with $\phi = 0.95$ (solid) and $\phi = 0.3$ (dashed), using the full history. (b)~Multiplicative seasonal $\mathrm{AR}(1) \times (1)_{12}$ with $\phi = 0.5$, $\Phi = 0.8$, using $\mathcal{I}_t = \sigma(Y_t)$ to illustrate non-monotonicity.}
\label{fig:profiles}
\end{figure}

\subsection{Complexity futility}

\begin{corollary}[Complexity futility]
At any horizon, the maximum achievable reduction in expected log loss relative to the unconditional baseline is $F(h; \mathcal{I}_t)$, regardless of the complexity of the forecasting method or the computational resources deployed.
\end{corollary}

Increasing model complexity can reduce only the approximation component $\mathbb{E}[D_{\mathrm{KL}}(p_h^\star \| q_h)]$; it cannot alter the irreducible component $H(Y_{t+h} \mid \mathcal{I}_t)$, which is a property of the data and the information set. At horizons where forecastability is small, even a perfect method offers only a negligible gain over the unconditional baseline under logarithmic loss. The corollary does not imply that forecasts cannot be improved; it implies that improvement requires enriching the information set, not the model class.

\section{Decomposition of Predictive Loss}

\begin{proposition}[Loss decomposition]
For any forecasting procedure $q_h(\cdot \mid \mathcal{I}_t)$ satisfying Assumption~1, the expected log loss separates as
\begin{equation}
\mathbb{E}[-\log q_h(Y_{t+h} \mid \mathcal{I}_t)] = \underbrace{H(Y_{t+h} \mid \mathcal{I}_t)}_{\text{irreducible}} + \underbrace{\mathbb{E}[D_{\mathrm{KL}}(p_h^\star \| q_h)]}_{\text{approximation}}.
\end{equation}
\end{proposition}

\begin{proof}
Direct consequence of Theorem~1: the expected log loss of any predictive distribution equals the conditional entropy plus its expected KL divergence from the true conditional distribution.
\end{proof}

The first term is irreducible and fixed by the joint distribution of the process and information set. The second is approximation error and can, in principle, be reduced by better modelling. This decomposition parallels the calibration--refinement factorisation in probability forecasting~\cite{murphy1973new}, but is exact under log loss rather than approximate under other scoring rules.

Equivalently, the expected log loss can be written as
\begin{equation}
\mathbb{E}[-\log q_h] = H(Y_{t+h}) - F(h; \mathcal{I}_t) + \mathbb{E}[D_{\mathrm{KL}}(p_h^\star \| q_h)].
\end{equation}
This form identifies two distinct levers for improvement: enriching the information set to increase $F(h; \mathcal{I}_t)$, and improving the forecasting method to reduce the KL term. Even a perfect method can eliminate only the approximation term; the irreducible component is fixed by the information structure of the data.

For a specified method~$q_h$, the realised reduction relative to the unconditional baseline, $X_q(h; \mathcal{I}_t) = H(Y_{t+h}) - \mathbb{E}[-\log q_h(Y_{t+h} \mid \mathcal{I}_t)]$, may be termed the exploitability. The ratio $\chi_q = X_q / F$ gives the fraction of available forecastability that the method achieves. Forecastability is a property of the process and information set; exploitability is method-dependent and can only be evaluated post-modelling. The distinction echoes Granger's~\cite{granger1999empirical} separation of predictability from achieved forecast performance.

\section{Horizons of Low Forecastability}

At horizons where forecastability is small, the irreducible floor is high. When $F(h; \mathcal{I}_t)$ is close to zero,
\begin{equation}
H(Y_{t+h} \mid \mathcal{I}_t) \approx H(Y_{t+h}).
\end{equation}
At such horizons, the best possible probabilistic forecast is already close to the unconditional baseline in expected log loss. Low forecastability does not imply that prediction is impossible, only that the maximum achievable improvement over the unconditional baseline is small under logarithmic loss.

Two classical inequalities make the consequences precise. Fano's inequality~\cite{fano1961transmission} gives a lower bound on classification error: for any predictor $\hat{Y}_{t+h}$ of a discrete $Y_{t+h}$ taking $M$ values,
\begin{align}
P(\hat{Y}_{t+h} \neq Y_{t+h}) &\geq \frac{H(Y_{t+h} \mid \mathcal{I}_t) - 1}{\log M} \notag\\
&= \frac{H(Y_{t+h}) - F(h; \mathcal{I}_t) - 1}{\log M},
\end{align}
so when forecastability is small, the minimum achievable error probability is high regardless of method. Pinsker's inequality~\cite{pinsker1964information} bounds the total variation distance between the conditional and marginal predictive distributions:
\begin{equation}
\|p_h^\star(\cdot \mid \mathcal{I}_t) - p(\cdot)\|_{\mathrm{TV}} \leq \sqrt{\tfrac{1}{2}\, F(h; \mathcal{I}_t)},
\end{equation}
so when forecastability is small, the two distributions are close: conditioning on the information set barely changes the predictive distribution.

This is not a failure of the forecasting method; it is a consequence of the information structure. The standard response to poor forecasting accuracy, namely increasing model complexity, is effective only if the approximation component is the binding constraint. When it is the irreducible component that dominates, no amount of additional complexity can help: when forecastability is near zero, all methods converge toward the same floor. The minimum description length principle~\cite{rissanen1978modeling} provides a complementary perspective: model complexity is justified only when it captures genuine regularity. At low-forecastability horizons, the data do not support complex descriptions, and the optimal code length approaches the marginal entropy~\cite{barron1998minimum}.

\section{Estimation}

The results above are stated in terms of the true forecastability $F(h; \mathcal{I}_t)$. In practice, this must be estimated from data. Several approaches exist. Histogram-based estimators partition the observation space but suffer from binning artefacts~\cite{darbellay1999estimation}. The standard non-parametric approach is the $k$-nearest-neighbour estimator of Kraskov, St\"{o}gbauer and Grassberger~\cite{kraskov2004estimating}, which avoids explicit density estimation and is consistent under regularity conditions~\cite{paninski2003estimation}. Gao et al.~\cite{gao2015efficient} extend the KSG framework to improve performance under strong dependence, and Kozachenko and Leonenko~\cite{kozachenko1987sample} provide the foundational entropy estimator on which KSG builds. The main practical constraint is the dimensionality of $\mathcal{I}_t$: in low-dimensional lag-based settings, the estimator is reliable for moderate sample sizes, but in higher-dimensional settings, finite-sample bias and variance grow and may make the diagnostic as demanding as the modelling it is intended to precede. In principle, richer information sets can only increase true forecastability, but in finite samples they may make estimated forecastability less reliable. In practice, the diagnostic is most tractable in the univariate lag-based setting where $\mathcal{I}_t = \sigma(Y_t, \ldots, Y_{t-p+1})$ with $p \leq 3$, series lengths $n \geq 200$, and horizons where the effective sample size $n - h$ remains sufficient for stable estimation. The companion empirical work~\cite{catt2026diagnostic} validates the diagnostic in this regime across more than 42,000 series. Permutation-based significance testing~\cite{good2005permutation,theiler1992testing}, in which the dependence between past and future is destroyed by random shuffling, provides a natural null for determining whether estimated forecastability at a given horizon differs from zero~\cite{catt2026diagnostic}. Fuller estimation methodology is developed in the companion empirical work~\cite{catt2026diagnostic}. Formal estimation theory---including finite-sample convergence rates, minimax bounds, and the interaction between lag dimension and sample size---is not developed here and remains an open direction.

\section{Implications for Forecasting Practice}

Under logarithmic loss and for the declared information set, forecasting is constrained by information before it is constrained by modelling. The forecastability profile provides a theoretical basis for identifying where predictive structure exists prior to selecting and fitting a forecasting model class. The framework is diagnostic rather than prescriptive: it characterises the information-limited structure of the forecasting problem, but does not by itself determine a forecasting model or decision rule.

The data processing inequality (Corollary~5) has direct implications for model design: any feature extraction, dimensionality reduction, or data compression applied to the information set can only reduce the forecastability ceiling. The finite-window result (Proposition~6) gives an exact error budget for lag truncation, connecting the theoretical framework to the practical reality that forecasters always work with finite histories.

In many modern forecasting workflows~\cite{makridakis2020m4,petropoulos2022forecasting}, model development already involves competing specifications, repeated validation, and sometimes ensembling; the practical role of forecastability is to constrain that wider search, not to replace modelling altogether~\cite{fildes2009effective}. When forecastability is high, model improvement can reduce the approximation component and is therefore productive. When forecastability is low, the achievable improvement is tightly bounded, and increasing complexity risks overfitting without meaningful gain. The forecastability profile provides a principled basis for making this distinction.

Regime shifts alter the underlying dependence structure~\cite{clements1998forecasting,hamilton1989new} and can render previously estimated forecastability profiles unrepresentative. Reliable re-estimation may require substantial post-shift data, so forecastability should be interpreted locally in such settings.

\subsection{Illustrative example: seasonal information structure}

To illustrate how the profile changes a structural decision that a single diagnostic would miss, consider a monthly series with strong annual seasonality~\cite{hyndman2021fpp}. Suppose that $F(h)$ is substantial for $h \in \{1, 2, 3, 4, 11, 12, 13\}$ and negligible for $h \in \{5, 6, 7, 8, 9, 10\}$. A single first-crossing diagnostic would report $h = 5$ as the point at which prediction becomes uninformative. The forecastability profile reveals that this conclusion is wrong: horizons 11--13 contain substantial predictive information because the seasonal dependence re-emerges at the annual lag, as formalised in Proposition~7.

Now suppose an external leading indicator is added that is correlated with the target at horizons 6--8 but uncorrelated elsewhere. The updated profile may show substantial forecastability at $h \in \{1, 2, 3, 4, 6, 7, 8, 11, 12, 13\}$, expanding the productive horizons from seven to ten. This demonstrates that the forecastability profile is not a fixed property of the target series but changes structurally when the declared information set is enriched.

\subsection{Extension to alternative losses}

For Gaussian processes, the forecastability profile has a direct connection to the MSE-based predictability familiar from the Wiener--Granger tradition. If $(Y_{t+h}, \mathcal{I}_t)$ is jointly Gaussian, the forecastability at horizon~$h$ reduces to
\begin{equation}
F(h; \mathcal{I}_t) = -\tfrac{1}{2}\log(1 - R_h^2),
\end{equation}
where $R_h^2$ is the population $R^2$ from regressing $Y_{t+h}$ on $\mathcal{I}_t$. Granger-causal predictability under squared-error loss is captured by $R_h^2$; forecastability under log loss is a monotone transformation of it. In the Gaussian case, the two order series identically: a series that is more forecastable under log loss is also more predictable under MSE. In the non-Gaussian case, the two can diverge because $F$ captures nonlinear dependence that $R_h^2$, being a second-moment quantity, misses entirely.

More generally, logarithmic loss occupies a distinguished position among proper scoring rules: in binary and finite-alphabet settings, its associated divergence upper-bounds the divergences induced by broad classes of smooth, proper, convex losses up to normalisation~\cite{gneiting2007strictly,painsky2020universality,ehm2016quantiles}. Low forecastability under log loss therefore implies low forecastability under these losses as well, but the implication is one-directional and may be loose at moderate forecastability levels. For practitioners who evaluate forecasts under squared-error loss, CRPS, or quantile loss, the exact equality proved here does not transfer: forecastability under log loss bounds but does not equal the maximum achievable improvement under those losses. The exact equality is specific to log loss; for other losses, the relationship is an inequality rather than an identity. For continuous-outcome forecasting under squared-error loss, related bounds may be developed through entropy-based inequalities, but these are weaker and more model-dependent. A full loss-specific treatment belongs more naturally to a rate-distortion framework~\cite{berger1971rate,cover2006elements} and lies beyond the scope of the present paper.

\section{Conclusion}

Under logarithmic loss, the mutual information between the future and the declared information set is exactly the maximum achievable reduction in expected loss. This paper has developed the consequences of that identity for forecasting: the forecastability profile as a horizon-resolved diagnostic, its behaviour under information-set compression (Corollary~5), the exact error budget for finite-window truncation (Proposition~6), the inheritance of periodicity (Proposition~7), and the decomposition of predictive loss into irreducible and approximation components, with the forecastability exploitation ratio $\chi_q$ providing a normalised measure of how much predictive structure a given method captures. The framework provides a theoretical foundation for assessing whether the declared information set contains sufficient predictive information at the horizon of interest, prior to committing to a model class. Companion empirical work~\cite{catt2026diagnostic} validates the diagnostic against realised forecast accuracy across more than 42,000 time series: auto-mutual information computed strictly from training data exhibits Spearman rank correlations of $-0.41$ to $-0.72$ with out-of-sample sMAPE for five of six temporal frequencies, and median forecast error declines monotonically from low to high forecastability terciles across all frequencies and probe models. Open extensions include loss-specific bounds via rate-distortion theory, the behaviour of the profile under non-stationarity and regime change, computable bounds on the information loss for specific compression schemes (e.g.\ principal component projections or feature maps), and integration with established forecast evaluation protocols~\cite{diebold1995comparing,tashman2000outof}. The framework extends directly to multivariate information sets: all results are stated in terms of sub-$\sigma$-algebras and hold without modification when $\mathcal{I}_t$ includes vector-valued histories, exogenous variables, or mixed-frequency data. The central implication is that forecastability is a measurable property of the data, not an emergent property of the model. It can be quantified before any model is built, and it bounds what every model can achieve.

\end{document}